\documentclass[aps,reprint,superscriptaddress]{revtex4-2}
\usepackage{comment}
\usepackage{graphicx} 
\usepackage{amsmath}
\usepackage{amsthm}
\usepackage{amssymb}
\usepackage{mathtools}
\usepackage{dsfont}
\usepackage{physics}
\usepackage{subcaption}
\usepackage{xcolor}
\usepackage{hyperref}

\newtheorem{theorem}{Theorem}

\newtheorem{corollary}{Corollary}[theorem]

\newtheorem{lemma}[theorem]{Lemma}

\newtheorem{proposition}[theorem]{Proposition}
\newtheorem{remark}[theorem]{Remark}

\newtheorem{theorem*}{Theorem}
\newtheorem{lemma*}[theorem*]{Lemma}
\newtheorem{proposition*}[theorem*]{Proposition}

\begin{document}

\title{A Representation-Theoretic Framework for Characterizing Barren Plateaus}

\author{Pedro Alcântara}
\email{pedro.antonio.alcantara@usp.br} 
\affiliation{International Institute of Physics, Federal University of Rio Grande do Norte, 59078-970, Natal, Brazil}
\author{Leandro Morais}
\email{leandro.silva@ifsc.usp.br}
\affiliation{Instituto de Física de São Carlos, Universidade de São Paulo, CP 369, 13560-970 São Carlos, SP, Brazil}
\author{Rafael Chaves}
\affiliation{International Institute of Physics, Federal University of Rio Grande do Norte, 59078-970, Natal, Brazil}

\date{\today}

\begin{abstract}
The scalability of variational quantum algorithms is fundamentally limited by the barren plateau effect, where the cost-function variance vanishes with system size, rendering optimization impractical. Recent Lie-algebraic approaches for deep parameterized have enabled a unified analytical understanding of this challenge but require either the initial state or the measurement observable to belong to the dynamical Lie algebra generated by the circuit. Here, we introduce a representation-theoretic framework under $2$-design hypothesis showing that variational quantum landscapes admit a natural decomposition into irreducible representation channels. This yields exact expressions and analytical bounds for the cost-function variance applicable to arbitrary initial states and observables, with previous Lie-algebraic results emerging as a special case. We illustrate the framework by analyzing the energy landscape of the one-dimensional ANNNI model for several circuit architectures, revealing trainability regimes inaccessible to existing methods. Our results establish a general representation-theoretic framework for analyzing variational quantum landscapes, substantially extending the analytical theory of barren plateaus.
\end{abstract}

\maketitle

\section{Introduction}

Variational quantum algorithms (VQAs) offer a versatile framework for leveraging quantum processors across quantum simulation, optimization, and machine learning \cite{peruzzo2014variational,farhi2014quantum,beckey2022variational,cerezo2021variational,sciorilli2025towards}. However, their scalability is threatened by barren plateaus (BPs), regimes in which loss differences become exponentially suppressed with increasing system size, rendering the optimization increasingly challenging \cite{sweke2020stochastic,sweke2020stochastic,stokes2020quantum,crooks2019gradients,arrasmith2022equivalence}. In this regime, resolving a useful descent direction from finite-shot measurements requires an exponential number of measurement samples, rendering optimization intractable \cite{wang2021noise,larocca2025barren,arrasmith2022equivalence,holmes2022connecting}. Understanding the structural mechanisms that govern the emergence of barren plateaus has therefore become a central problem in variational quantum computing \cite{de2025lie,cerezo2025does,kazi2025analyzing,ragone2022representation}.

A major step toward a unified understanding of barren plateaus was provided by dynamical Lie algebra (DLA) methods 
\cite{ragone2024lie,ragone2024lie,goh2025lie,fontana2024characterizing,wiersema2024classification}.
For sufficiently deep parametrized circuits, the DLA generated by the circuit determines a compact Lie group over which the circuit explores its state space, allowing the cost-function variance to be expressed in terms of algebraic properties of the circuit, the initial state, and the observable. This framework unifies several previously distinct mechanisms for the emergence of barren plateaus, including circuit expressibility, properties of the input state, and the structure of the cost observable~\cite{ragone2024lie}. It also clarifies why highly expressive circuits generically develop exponentially flat landscapes, while restricted ansatze can retain trainable directions. However, the resulting formulas rely on a crucial structural assumption: either the initial state or the observable must belong to the dynamical Lie algebra generated by the circuit. While this assumption holds in certain benchmark scenarios, it excludes many physically relevant problems and leaves a broad class of variational landscapes uncharacterized. 

While the DLA characterizes the symmetries of a given circuit, it is a subspace strictly contained within the much larger operator space $\mathcal{B}(\mathcal{H})$. Therefore, restricting the initial state or measurement observable to the DLA leaves out a broad and physically relevant class of operators, those containing components that transform outside the adjoint representation of the DLA. This occurs, for example, for observables containing higher-order correlations, highly entangled states, encoded states, and Hamiltonians whose operator content extends beyond the algebra generated by the circuit. Recent work has demonstrated that the DLA restriction can indeed be relaxed in specific settings, such as parametrized matchgate circuits, where the relevant structure is captured by generalized representation modules rather than by the Lie algebra itself \cite{diaz2023showcasing}. These results point toward a broader representation-theoretic description, but a general framework for arbitrary compact Lie groups, initial states, and observables remains desirable.

In this work, we show that this limitation is not intrinsic to the problem, but rather to the mathematical language used to describe it. We regard the variational quantum landscape---the cost function defined on the parameter manifold, or equivalently on the orbit of the underlying compact Lie group---as the fundamental object of study. This perspective shifts the emphasis from operators in the dynamical Lie algebra to functions on group orbits, revealing that the characterization of landscape fluctuations is naturally a problem of harmonic analysis. By decomposing the landscape into irreducible representation channels, we derive exact expressions together with analytical upper and lower bounds for the cost-function variance that are valid for arbitrary initial states and observables. Rather than replacing the existing Lie-algebraic theory, our formulation contains it naturally as a particular representation-theoretic sector, with the previously known DLA formulas emerging as a special case. 

We illustrate the framework through the energy landscapes of the one-dimensional axial next-nearest-neighbor Ising (ANNNI) model \cite{selke1988annni,Suzuki,canabarro2019unveiling,morais2026distinguishing} for three circuit architectures with qualitatively different symmetry structures. The fully expressive $\mathfrak{su}(2^n)$ architecture exhibits the familiar exponentially suppressed variance and barren plateau. In contrast, the local product architecture generated by $\mathfrak{su}(2)^n$ has an extensive variance, while a highly constrained architecture generated by $\mathfrak{u}(1)^{2n-1}$ admits a system-size-independent lower bound on the variance. These regimes are analytically derived for separable stabilizer, separable nonstabilizer, and maximally entangled initial states, and are corroborated by numerical simulations. Taken together, the examples demonstrate that the representation-theoretic structure of the circuit determines how fluctuations are distributed across the landscape, providing a general framework for analyzing trainability beyond the restrictions of dynamical Lie algebra membership.

\section{Harmonic Analysis of Variational Quantum Landscapes}\label{sec:framework}

Our goal is to derive analytical expressions for the mean and variance of the cost function associated with deep parameterized quantum circuits. Throughout this work, we assume that the circuit ensemble forms a unitary $2$-design over a compact Lie group $G$. Under this assumption, averages of linear and quadratic functionals over circuit parameters coincide with their Haar averages over $G$, allowing the trainability problem to be reformulated in terms of harmonic analysis on compact Lie groups.

Rather than working directly with the circuit parameters, our approach exploits the geometry of the orbit generated by the initial state. The central idea is to decompose the cost function into contributions associated with the irreducible representations of the underlying group. As we shall see, this decomposition naturally yields a corresponding decomposition of the variance, leading to a general expression valid for arbitrary initial states and observables.

The results in this section are stated without proof in order to keep a clear discussion of their meaning. The proofs are essentially based on Schur's Orthogonality Relations and can be found in Appendix \ref{sec:proof}. For a more detailed discussion about Lie Algebras, Representation Theory and Harmonic Analysis, we refer to \cite{fulton2013representation,humphreys1972representation,folland}.

Let $\mathcal H_n=(\mathbb C^2)^{\otimes n}$ denote the Hilbert space of an $n$-qubit system, and let $\mathcal B(\mathcal H_n)$ be the space of linear operators acting on $\mathcal H_n$, endowed with the Hilbert--Schmidt inner product
\begin{equation}
\ip{A_1}{A_2}=\mathrm{Tr}(A_1^\dagger A_2).
\end{equation}
Given an initial pure state $\rho$, an observable $O$, and a parameterized quantum circuit
$\mathcal U(\boldsymbol\theta)$, the cost function can be written as
\begin{equation}
\ell_{\boldsymbol\theta}(\rho,O) = \ip{\mathcal U(\boldsymbol\theta)\rho\mathcal U^\dagger(\boldsymbol\theta)}{O}\,.
\label{cost_func}
\end{equation}

Since we work in the deep-circuit regime, expectation values with respect to the circuit parameters are replaced by Haar averages over the compact Lie group generated by the circuit. Recall that the group of unitary operators on $\mathcal H_n$ is the compact Lie group $U(\mathcal H_n) \equiv U(2^n)$, and let $\varphi : G\to U(2^n)$ be a unitary representation of a compact Lie group $G$, which induces a representation on $\mathcal B(\mathcal H_n)$ by
\begin{equation}\label{op_action}
    \mathcal B(\mathcal H_n)\ni A\mapsto A^g = \varphi(g)A\varphi(g)^\dagger \ , \ \ \forall \, g \in G\, .
\end{equation}
For the $G$-orbit $\mathcal O$ generated by $\rho$, we set $\beta_{\rho, O}\in C^\infty(\mathcal O)$,
\begin{equation}
    \mathcal O \ni \rho' \mapsto \beta_{\rho, O}(\rho') = \ip{\rho'}{O} = \tr(\rho'O)\, ,
\end{equation}
and its lift $\widetilde \beta_{\rho, O}\in C^\infty(G)$,
\begin{equation}
    \widetilde \beta_{\rho, O}(g) = \beta_{\rho, O}(\rho^g).
\end{equation}
Consider the quantities
\begin{equation}\label{beta_mean}
    \langle \beta_{\rho, O} \rangle: = \int_{\mathcal O} \beta_{\rho, O}(\rho')d\rho' = \int_G \widetilde \beta_{\rho, O}(g)dg\, ,
\end{equation}
\begin{equation}\label{sq_beta_mean}
    \langle \beta_{\rho, O}^2 \rangle:= \int_{\mathcal O} \beta_{\rho, O}(\rho')^2d\rho' = \int_G \widetilde \beta_{\rho, O}(g)^2dg\, ,
\end{equation}
where the integrals are taken with respect to the normalized Haar measure of $G$ and the induced invariant measure on $\mathcal O$.

Let $\mathfrak g$ be the Lie algebra of $G$. We suppose the image of the naturally induced representation $\varphi_\ast:\mathfrak g \to \mathcal B(\mathcal H_n)$ is a DLA describing the deep parameterized quantum circuit depicted in \eqref{cost_func}. Thus, for any observable $O$, we have
\begin{equation}\label{mean_loss}
    \mathbb E_{\vb*\theta}[\ell_{\vb*\theta}(\rho, O)] = \langle \beta_{\rho, O} \rangle \, , 
\end{equation}
\begin{equation}\label{var_loss}
    \mathrm{Var}_{\vb*\theta}[\ell_{\vb*\theta}(\rho, O)] = \langle\beta_{\rho, O}^2\rangle - \langle \beta_{\rho, O}\rangle^2\, .
\end{equation}

Now, let
\begin{equation}\label{op_inv_decomp}
    \mathcal B(\mathcal H_n) \equiv \bigoplus_{\zeta \in \mathfrak{CG}(\varphi)} \mathcal H_\zeta\otimes \mathbb C^{\mathfrak m(\zeta)},
\end{equation}
be the decomposition of \eqref{op_action} into unitary irreps $\zeta:G\to U(\mathcal H_\zeta)$, where distinct elements of $\mathfrak{CG}(\varphi)$ are nonequivalent and $\mathfrak{m}(\zeta)\in \mathbb N$ denotes the multiplicity \footnote{The notation $\mathfrak{CG}(\varphi)$ refers to the fact that, if $\varphi$ is an irrep, such decomposition is an instance of a Clebsch-Gordan series for the $G$-representation $\varphi\otimes \varphi^\ast$, where $\varphi^\ast$ is the dual representation of $\varphi$.}. In particular, since $\mathds 1$ is fixed by $G$, $\mathfrak{CG}(\varphi)$ contains the trivial representation $\tau(g) = 1$ on $\mathcal H_\tau \simeq \mathbb C$ at least once. For each $\zeta\in \mathfrak{CG}(\varphi)$, let
\begin{equation}\label{pi_zeta}
    \pi_\zeta:\mathcal B(\mathcal H_n)\to \mathcal H_\zeta\otimes \mathbb C^{\mathfrak m(\zeta)},
\end{equation}
be the natural projection.

The decomposition \eqref{op_inv_decomp} naturally separates the operator space into invariant representation sectors, which organizes the statistical properties of the cost function so that its first two momenta can be computed by means of the identification of the contribution from each such sector. Indeed, using the decomposition \eqref{op_inv_decomp}-\eqref{pi_zeta}, the lifted cost function naturally splits into 
a finite Fourier series, with orthogonal 
contributions from each irreducible representation, as \footnote{By construction, the decomposition in \eqref{op_inv_decomp} is orthogonal with respect to the Hilbert-Schmidt inner product, meaning the maps $\pi_\zeta$ are, by definition, orthogonal projectors. As such, the identity $\ip{\pi_\zeta(A)}{B} = \ip{\pi_\zeta(A)}{\pi_\zeta(B)}$ holds for any operators $A, B \in \mathcal{B}(\mathcal{H}_n)$. However, we purposefully keep $O$ instead of its projections to maintain a clearer notation in the subsequent proofs.}
\begin{equation}\label{Beta_projection}
    \widetilde \beta_{\rho, O}(g) = \sum_{\zeta \in \mathfrak{CG}(\varphi)}\ip{\pi_\zeta(\rho)^g}{O}\, .
\end{equation}

\begin{remark}
Of course, a splitting equivalent to \eqref{Beta_projection} holds for the cost function $\beta_{\rho, O}$ itself, but we shall use its lift in the following calculations as the dependence on $g \in G$ makes the formulae more presentable.    
\end{remark}

The upcoming results states how each invariant representation $\zeta\in \mathfrak{CG}(\varphi)$ determines independent contributions to the first and the second moment of $\ell_{\vb*\theta}(\rho, O)$, and that these contributions sum up to the total values of the momenta of interest. 

To begin with, the next lemma shows that the average cost function depends exclusively on the trivial representation.

\begin{lemma}\label{lemma:mean_beta}
\begin{equation*}
    \langle \beta_{\rho, O}\rangle = \ip{\pi_\tau(\rho)}{\pi_\tau(O)} = \tr(\pi_\tau(\rho)\pi_\tau(O))\, .
\end{equation*}
\end{lemma}

Lemma \ref{lemma:mean_beta} reflects the fact that averaging over the group action removes every component of the state and observable except those on which the $G$-action is trivial. In other words, the group average acts as a projector onto the subspace of $\mathcal B(\mathcal H_n)$ fixed by $G$, so that only the fixed parts of $\rho$ and $O$ contribute to the mean value of the cost function. 

The next proposition shows that, under the assumptions commonly considered in the Lie-algebraic theory of barren plateaus, this expression simplifies considerably. In particular, for irreducible representations and traceless observables, the mean vanishes identically, so the variance is entirely determined by the second moment.

\begin{proposition}\label{prop2}
    Under the hypothesis of $2$-design, if $\varphi$ is an irrep and $\tr(O) = 0$, then $\mathbb E_{\vb*\theta}[\ell_{\vb*\theta}(\rho,O)] = 0$ and $\mathrm{Var}_{\vb*\theta}[\ell_{\vb*\theta}(\rho, O)]=\langle\beta_{\rho, O}^2\rangle$.
\end{proposition}

Our objective, however, is to derive a general expression for the variance without assuming either irreducibility of the representation or vanishing mean. To this end, we now analyze the second moment by exploiting the decomposition \eqref{Beta_projection}, so we define
\begin{equation}\label{psi_comp_mean_beta_sq}
\begin{aligned}
    \mathcal C_\zeta(\rho, O) & :=\int_G\ip{\pi_\zeta(\rho)^g}{O}\ip{O}{\pi_\zeta(\rho)^g}dg\\
    & = \int_G|\tr(\pi_\zeta(\rho)^g\pi_\zeta(O))|^2dg,
\end{aligned}
\end{equation}
for each $\zeta \in \mathfrak{CG}(\varphi)$. As the following lemma shows, it is precisely $\mathcal C_\zeta(\rho, O)$ that quantifies the contribution of the irreducible representation $\zeta$ to the second moment of the cost function (see Fig. \ref{fig:sector}).

\begin{lemma}\label{lemma:mean_beta_sq}
    \begin{equation*}
    \begin{aligned}
        \langle \beta_{\rho, O}^2\rangle & = \sum_{\zeta\in\mathfrak{CG}(\varphi)}\mathcal C_\zeta(\rho, O)\, .
    \end{aligned}
    \end{equation*}
\end{lemma}

The next step is to understand the relevant properties of $\mathcal C_{\zeta}(\rho, O)$.

\begin{figure}[t!]
    \includegraphics[width=1\linewidth]{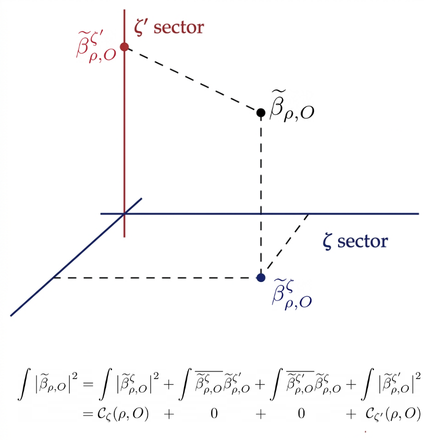}
    \caption{Geometric interpretation of the cost function variance decomposition. The function $\widetilde{\beta}_{\rho,O}$ is projected onto distinct irreducible representation sectors, cf. \eqref{Beta_projection}. Upon integration over the group, cross-correlation terms vanish due to Schur's orthogonality, reducing the second moment to a sum of independent, non-negative contributions as established in Lemma \ref{lemma:mean_beta_sq}.}
    \label{fig:sector}
\end{figure}

We begin with the trivial representation. It plays a distinguished role because it captures the invariant component responsible for the mean value of the cost function. As the next lemma shows, its contribution admits a particularly simple expression.

\begin{lemma}\label{lemma:C_tau}
For the trivial representation,
\begin{equation*}
    \mathcal C_\tau(\rho, O) = \tr(\pi_\tau(\rho)\pi_\tau(O))^2\, . 
\end{equation*}
\end{lemma}

Lemma \ref{lemma:C_tau} confirms that the contribution of the trivial representation to the second moment is precisely the square of the mean established in Lemma \ref{lemma:mean_beta}. Consequently, when computing the variance, this contribution is exactly canceled by the subtraction of $\langle\beta_{\rho,O}\rangle^2$. The remaining irreducible representations therefore account entirely for the fluctuations of the cost function.

For non-trivial representations, the following lemma establishes that $\mathcal C_\zeta(\rho, O)$ is always nonnegative and provides a universal upper bound determined solely by the projections of the state and observable onto the corresponding representation sector.

\begin{lemma}\label{lemma:C_zeta_natural}
For any $\zeta \in \mathfrak{CG}(\varphi)$, we have 
\begin{equation*}
    0\le \mathcal C_\zeta(\rho, O) \le \dfrac{1}{\dim \zeta}\norm{\pi_\zeta(\rho)}^2\norm{\pi_\zeta(O)}^2\, ,
\end{equation*}
where the rightmost inequality is an actual equation if $\zeta$ is multiplicity free.
\end{lemma}

Lemma \ref{lemma:C_zeta_natural} shows that the representation contributions are tightly controlled by the norms of the projected state and observable, with the dimension of the irreducible representation setting the natural scale. In the multiplicity-free case, the upper bound is saturated, as a consequence of Schur's Orthogonality Relations. These properties will allow us to derive both exact variance formulas and general analytical bounds in our main theorem.

The definition of $\mathcal C_\zeta(\rho,O)$ is intrinsic and independent of any particular basis. For practical calculations, however, it is useful to obtain an explicit expression in terms of the coordinates of $\rho$ and $O$ within each irreducible representation sector. To do so, it is necessary to specify a resolution for the multiplicities in the invariant decomposition \eqref{op_inv_decomp}. There are various approaches to deal with such degeneracies, with no natural universal method. The choice must be guided by the simplification of the computations to be made, as illustrated by our third example in Sec. \ref{sec:applications}. For now, we'll follow abstractly, supposing that such a resolution is taken, but with no specification of how it is done: denote by $\{e_j:1\le j \le \mathfrak m(\zeta)\}$ the canonical basis of $\mathbb C^{\mathfrak m(\varphi)}$, and let $\{u(\zeta;k):1\le k \le \dim \zeta\}$ be a standard orthonormal basis for $\mathcal H_\zeta$, so we get an orthonormal basis of $\mathcal B(\mathcal H_n)$ comprised by the operators
\begin{equation}\label{inv_basis}
    u(\zeta,j;k) \equiv u(\zeta;k)\otimes e_j\, ,
\end{equation}
for which
\begin{equation}\label{proj_psi_decomp}
\begin{aligned}
    \pi_\zeta(A) & = \sum_{j=1}^{\mathfrak m(\zeta)}\sum_{k=1}^{\dim\zeta}A_k^{(\zeta,j)} u(\zeta,j;k) \, ,\\
      \mbox{with} & \ A_k^{(\zeta,j)} = \ip{u(\zeta,j;k)}{A}\, .
\end{aligned}
\end{equation}
The following lemma provides a closed formula using the above coordinates.

\begin{lemma}\label{lemma:C_zeta_basis}
For \eqref{inv_basis}-\eqref{proj_psi_decomp},
    \begin{equation*}
        \mathcal C_\zeta(\rho, O) = \dfrac{1}{\dim \zeta}\sum_{k,l=1}^{\dim\zeta}\left|\sum_{j=1}^{\mathfrak m(\zeta)}O^{(\zeta,j)}_{k}\overline{\rho^{(\zeta,j)}_l}\right|^2\, .
    \end{equation*}
\end{lemma}

Lemma~\ref{lemma:C_zeta_basis} reveals how each representation contribution depends on the overlap between the projections of the state and observable onto the corresponding irreducible representation. In particular, the multiplicity indices couple different copies of the same irreducible representation, whereas inequivalent representations remain completely independent.

We are now ready to assemble the previous ingredients into a general expression for the second moment. Lemma \ref{lemma:mean_beta_sq} provides the fundamental decomposition of the second moment into independent irreducible representation channels. Moreover, the upper bound in Lemma \ref{lemma:C_zeta_natural} shows that the contribution of each channel is controlled solely by the norms of the projected state and observable together with the dimension of the corresponding representation. This decomposition forms the basis of our general variance formula.

Combining the characterization of the mean with the decomposition of the second moment yields the central result of this work. The theorem below provides a complete representation-theoretic description of the mean and variance of the cost function for arbitrary initial states and observables under the $2$-design hypothesis.

\begin{theorem}\label{theo:mean_var_2-design}
    Under the hypothesis of $2$-design, the mean and the variance of $\ell_{\vb*\theta}(\rho, O)$ are given by
    \begin{equation*}
        \mathbb E_{\vb*\theta}[\ell_{\vb*\theta}(\rho,O)] = \tr(\pi_\tau(O)\pi_\tau(\rho))\, ,
    \end{equation*}
    \begin{equation*}
    \begin{aligned}
        \mathrm{Var}_{\vb*\theta}[\ell_{\vb*\theta}(\rho, O)]& = \sum_{\substack{\zeta\in \mathfrak{CG}(\varphi)\\ \zeta\ne \tau}}\mathcal C_\zeta(\rho, O)\\
        & \le \sum_{\substack{\zeta\in \mathfrak{CG}(\varphi)\\ \zeta\ne \tau}}\dfrac{1}{\dim\zeta}\norm{\pi_\zeta(\rho)}^2\norm{\pi_\zeta(O)}^2\, .
    \end{aligned}
    \end{equation*}
\end{theorem}

Theorem \ref{theo:mean_var_2-design} is the main result of this work, showing that the cost function admits a remarkably simple structure: the trivial representation contributes exclusively to the mean, while every nontrivial irreducible representation contributes independently to the fluctuations. In this sense, the variance is completely organized by the representation-theoretic decomposition of the operator space. This result considerably extends previous Lie-algebraic analyses by removing any assumption that either the initial state or the observable belongs to the dynamical Lie algebra.

An immediate consequence of the theorem is that each representation contribution provides a rigorous lower bound on the total variance.

\begin{corollary}\label{cor:var_ge_C}
    Under the hypothesis of $2$-design, for any nontrivial $\zeta\in \mathfrak{CG}(\varphi)$,
    \begin{equation*}
        \mathrm{Var}_{\vb*\theta}[\ell_{\vb*\theta}(\rho, O)]\ge \mathcal C_\zeta(\rho, O)\, .
    \end{equation*}
\end{corollary}

This simple observation is particularly useful in applications. Rather than evaluating the full variance, it is often sufficient to identify a single nontrivial representation with $\mathcal C_\zeta(\rho,O)>0$ in order to certify a nonvanishing variance.

Combining Theorem~\ref{theo:mean_var_2-design} with the explicit expression of Lemma~\ref{lemma:C_zeta_basis} yields the following coordinate representation of the variance.

\begin{corollary}\label{cor:var_basis}
Under the hypothesis of $2$-design, for the decomposition \eqref{inv_basis}-\eqref{proj_psi_decomp}, we have explicitly
    \begin{equation*}
    \begin{aligned}
        & \mathrm{Var}_{\vb*\theta}[\ell_{\vb*\theta}(\rho, O)] =\sum_{\substack{\zeta\in \mathfrak{CG}(\varphi)\\\zeta\ne \tau}}\dfrac{1}{\dim\zeta}\sum_{k,l=1}^{\dim\zeta}\left|\sum_{j=1}^{\mathfrak m(\zeta)}O^{(\zeta,j)}_{k}\overline{\rho^{(\zeta,j)}_l}\right|^2.
    \end{aligned}
    \end{equation*}
\end{corollary}

This expression makes the role of the multiplicity spaces completely explicit and is particularly convenient for analytical calculations and numerical implementations. 

It is worthy remarking that it shares structural similarities with the results presented in \cite[Theorem 1]{ragone2024lie}. To formally recover their result, let us consider the restrictive scenario where the observable $O$ belongs to complexification of the dynamical Lie algebra, explicitly $O \in i\varphi_\ast(\mathfrak{g})$ \footnote{Note that the image $\varphi_\ast(\mathfrak g)$ is a \emph{real} algebra comprised by skew-Hermitian operators, hence the imaginary unit and the complexification.}. Recall that the Lie algebra $\mathfrak{g}$ of a compact Lie group is reductive, meaning it admits a direct sum decomposition into ideals given by
\begin{equation}
\label{reduc_decomp}\mathfrak{g} = \mathfrak{g}_0 \oplus \mathfrak{g}_1 \oplus \cdots \oplus \mathfrak{g}_d\, ,
\end{equation}
where $\mathfrak{g}_0$ is an Abelian factor (the center), and $\mathfrak{g}_j$ is a simple Lie algebras for $1 \le j \le d$. Because of this reductive structure, the operator $O \in i\varphi_\ast(\mathfrak{g})$ can be uniquely decomposed as a sum
\begin{equation}
    O = \sum_{j=0}^d P_j(O)\, ,
\end{equation}
with $P_j$ being the orthogonal projection from $\mathcal B(\mathcal H_n)$ to the complexification of $\varphi_\ast(\mathfrak{g}_j)$. Under the adjoint action of $G$, the center $\mathfrak{g}_0$ corresponds to trivial representations, while each simple ideal $\mathfrak{g}_j$ carries as a distinct irreducible representation space. By this mutually non equivalence, for $1\le j \le d$, the irrep $\xi_j$ on $\mathfrak g_j$ appears with multiplicity exactly $1$ within $\mathfrak{g}$. This multiplicity-free structure guarantees that $\pi_{\xi_j}(O) = P_j(O)$, allowing the multiplicity sums in our general theorem to collapse. Consequently, denoting by $\mathcal P_j(A)$ the \emph{$\mathfrak g_j$-purity} of any Hermitian operator $A\in \mathcal B(\mathcal H_n)$,
\begin{equation}
    \mathcal P_j(A) := \norm{P_j(A)}^2\, ,
\end{equation}
the theorem in Ref. \cite{ragone2024lie} emerges naturally as a straightforward corollary of our Theorem \ref{theo:mean_var_2-design}. This demonstrates that the representation-theoretic approach developed here strictly extends the existing Lie-algebraic theory.

\begin{corollary}\label{cor:rag}
    Under the hypothesis of $2$-design, if $\rho\in i\varphi_\ast(\mathfrak g)$ or $O\in i\varphi_\ast(\mathfrak{g})$, then
    \begin{equation*}
        \mathbb E_{\vb*\theta}[\ell_{\vb*\theta}(\rho,O)] = \tr(P_0(O)P_0(\rho))
    \end{equation*}
    and
    \begin{equation*}
        \mathrm{Var}_{\vb*\theta}[\ell_{\vb*\theta}(\rho, O)] = \sum_{j=1}^d\dfrac{\mathcal P_j(\rho)\mathcal P_j(O)}{\dim \mathfrak{g}_j}\, .
    \end{equation*}
\end{corollary}

Corollary~\ref{cor:rag} shows that the Lie-algebraic framework corresponds to the particular situation in which only the irreducible representations associated with the simple ideals of the dynamical Lie algebra contribute. Our theorem therefore extends this picture to arbitrary representation sectors, which become essential whenever the state and the observable possess components outside the dynamical Lie algebra.

Although Theorem~\ref{theo:mean_var_2-design} provides an exact decomposition of the variance, evaluating the contributions of the representation may not always be practical. It is therefore useful to derive general upper bounds depending only on intrinsic properties of the observable. To advance in this direction, recall the operator square root $|A|:=\sqrt{A^\dagger A}$ and the Schatten $p$-norm
\begin{equation}
    \norm{A}_p = \left(\tr(|A|)^p\right)^{1/p}\, , \ p\in [1,\infty)\, ,
\end{equation}
with $\norm{A}_\infty$ being the highest eigenvalue of $|A|$. In particular, if we diagonalize the observable $O$, then $|O|$ is the diagonal operator whose entries correspond to the absolute values of the respective entries of $O$, and $\norm{O}_2$ coincides with the Hilbert-Schmidt norm $\norm{O}$ we've used so far. The next lemma establishes an upper bound in terms of Schatten norms.

\begin{lemma}\label{lemma:bound_sch_norm}
If $\varphi$ is an irrep, then
    \begin{equation*}
        \langle \beta_{\rho,O}^2\rangle\le \dfrac{\norm{O}_\infty\norm{O}_1}{2^n}\, .
    \end{equation*}
\end{lemma}

Unlike the previous bounds, which depend explicitly on the representation-theoretic decomposition, Lemma~\ref{lemma:bound_sch_norm} expresses the second moment solely through Schatten norms of the observable. This makes the result directly applicable even when the irreducible decomposition is not explicitly available.

Combining the previous lemma with the expression for the mean immediately yields a representation-independent upper bound on the variance.

\begin{theorem}\label{theo:bound_sch_norm}
    Under the hypothesis of $2$-design, if $\varphi$ is an irrep, then the variance of $\ell_{\vb*\theta}(\rho, O)$ satisfies
    \begin{equation}
        \mathrm{Var}_{\vb*\theta}[\ell_{\vb*\theta}(\rho, O)]\le \dfrac{\norm{O}_\infty\norm{O}_1-\tr(O)^2}{2^n}\, .
    \end{equation}
\end{theorem}

Theorem~\ref{theo:bound_sch_norm} complements our representation-theoretic analysis by providing a universal estimate that depends only on spectral properties of the observable. Although generally less informative than the exact decomposition of Theorem~\ref{theo:mean_var_2-design}, it offers a simple analytical bound that can be applied with less knowledge of the underlying representation structure than the required by Theorem \ref{theo:mean_var_2-design}.

\begin{remark}
The bound in Theorem~\ref{theo:bound_sch_norm} is expressed in terms of the Schatten $1$- and $\infty$-norms because this yields the sharpest estimate obtained through our approach. If desired, the bound can be reformulated using other Schatten norms by invoking standard norm inequalities, for example
\begin{equation}\label{ineq_norms}
    \norm{O}_1
    \le
    2^{n/2}\norm{O}_2
    \le
    2^n\norm{O}_\infty,
\end{equation}
or any other equivalence relation between norms. We nevertheless retain the formulation of Theorem~\ref{theo:bound_sch_norm} since it preserves the strongest estimate while allowing users to adapt the result to the norm most convenient for a given application.
\end{remark}

\begin{figure*}[t!] 
    \centering
        \includegraphics[width=\textwidth]{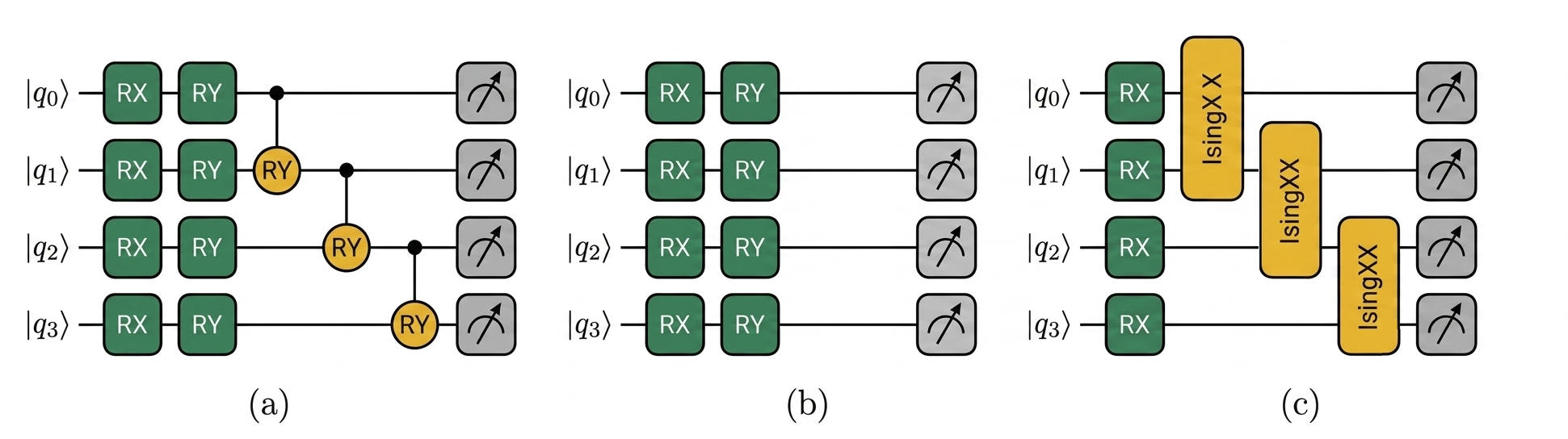}
    \caption{Parameterized building blocks for quantum circuit architectures (ansätze) used for the 1D ANNNI model energy minimization. In general, a number of layers of each building block are combined. (a) Hardware-Efficient Ansatz (HEA) with full entangling capabilities, generating the DLA $\mathfrak{su}(2^n)$. (b) Local separable ansatz with independent single-qubit operations, generating the DLA $\mathfrak{su}(2)^n$. (c) Structured ansatz featuring nearest-neighbor interactions, generating the DLA $\mathfrak{u}(1)^{2n-1}$.}
    \label{fig:tres_figuras_topo}
\end{figure*}

\section{Trainability Regimes in Variational Quantum Circuits}\label{sec:applications}
The representation-theoretic framework developed in the previous section provides a unified analytical description of variational quantum landscapes in terms of their decomposition into irreducible representation channels. Beyond recovering existing Lie-algebraic results, it enables the analytical characterization of optimization landscapes that lie outside the scope of previous approaches. In this section, we illustrate the framework through three representative variational ansatze whose dynamical Lie algebras span markedly different symmetry classes, ranging from fully expressive circuits to highly constrained architectures. Together, these examples demonstrate how the symmetry of the variational ansatz governs the structure of the landscape and, consequently, its trainability. Throughout this section, $X_j$, $Y_j$ and $Z_j$ denote, respectively the Pauli operators $\sigma_1$, $\sigma_2$ and $\sigma_3$ acting on the $j$-th qubit.

As a benchmark, we consider the variational minimization of the one-dimensional axial next-nearest-neighbor Ising (ANNNI) Hamiltonian \cite{selke1988annni,Suzuki,canabarro2019unveiling,morais2026distinguishing}
\begin{equation}\label{ANNNI_ham}
    H = -\sum_{j=1}^{n-1}Z_jZ_{j+1}
    +\kappa\sum_{j=1}^{n-2}Z_jZ_{j+2}
    -\gamma\sum_{j=1}^{n}X_j,
\end{equation}
where $\kappa,\gamma\ge0$. The ANNNI model provides an ideal testbed because it combines local fields with competing nearest and next-nearest-neighbor interactions, allowing different irreducible representation channels to contribute to the landscape. 

To highlight the role played by the initial state independently of the circuit symmetry, we consider the completely separable stabilizer state
\begin{equation}\label{state_s}
\begin{cases}
\ket{\psi_s}=\ket0^{\otimes n},\\[1ex]
\rho_s=\ketbra{\psi_s}{\psi_s}
=\dfrac1{2^n}\left(\mathds1+\sigma_3\right)^{\otimes n}
\end{cases},
\end{equation}
the completely separable non-stabilizer state
\begin{equation}\label{state_s2}
\begin{cases}
\ket{\psi_m}=\dfrac{1}{2^{n/2}}(\ket{0}+e^{i\pi/4}\ket{1})^{\otimes n},\\[1ex]
\rho_m=\ketbra{\psi_m}{\psi_m}
=\dfrac1{2^n}\left(\mathds1+\dfrac{1}{\sqrt{2}}(\sigma_1+\sigma_2)\right)^{\otimes n}
\end{cases},
\end{equation}
and the paradigmatic entangled Greenberger--Horne--Zeilinger (GHZ) state
\begin{equation}\label{state_e}
\begin{cases}
\ket{\psi_e}
=\dfrac1{\sqrt2}
\left(\ket0^{\otimes n}+\ket1^{\otimes n}\right),\\[1ex]
\rho_e
=\dfrac1{2^{n+1}}
\Big(
(\mathds1+\sigma_3)^{\otimes n}
+
(\mathds1-\sigma_3)^{\otimes n}\\
\hspace{7em}
+
(\sigma_1+i\sigma_2)^{\otimes n}
+
(\sigma_1-i\sigma_2)^{\otimes n}
\Big)
\end{cases}.
\end{equation}

Before analyzing each ansatz individually, it is useful to establish a general benchmark. First, we have that $\mathrm{tr}(H)=0$. In addition, the Hilbert-Schmidt norm of the ANNNI Hamiltonian is
\begin{equation}\label{2-norm_H}
\norm{H}
=
2^{n/2}
\sqrt{(1+\kappa^2+\gamma^2)n-1-2\kappa^2},
\end{equation}
while
\begin{equation}\label{infty-norm_H}
\norm{H}_{\infty}
\le
n(1+\kappa+\gamma)-(1+2\kappa).
\end{equation}
Therefore, under the $2$-design assumption, if the ansatz induces an irreducible representation, Theorem \ref{theo:mean_var_2-design} (or, Proposition \ref{prop2} for that matter) gives
\begin{equation}
    \mathbb E_{\vb*\theta}[\ell_{\vb*\theta}(\rho, H)] = \tr(\pi_\tau(H)\pi_\tau(\rho)) = 0
\end{equation}
for any initial state $\rho$. Additionally, under the same hypothesis, Theorem \ref{theo:bound_sch_norm} assumes the form
\begin{equation}
\begin{aligned}
    \mathrm{Var}_{\boldsymbol\theta}[\ell_{\boldsymbol\theta}(\rho,H)]\le \dfrac{\norm{H}_\infty \norm{H}_1-\tr(H)^2}{2^n}
\end{aligned}
\end{equation}
Now, we can use \eqref{ineq_norms} to substitute $\norm{H}_1$ by $\norm{H}_2$, so that $\tr(H) = 0$ and \eqref{2-norm_H}-\eqref{infty-norm_H} yields
\begin{equation}
\begin{aligned}
    \mathrm{Var}_{\boldsymbol\theta}[\ell_{\boldsymbol\theta}(\rho,H)] & \le \dfrac{\norm{H}_\infty\norm{H}}{2^{n/2}} \in \mathcal O(n^{3/2})\, ,
\end{aligned}
\end{equation}
an upper bound with power scaling of degree $3/2$.

The three circuit architectures considered below illustrate three qualitatively distinct trainability regimes (see Fig. \ref{fig:tres_figuras_topo}):

\begin{itemize}
\item[(i)] \textbf{Fully expressive circuits}, whose DLA is $\mathfrak{su}(2^n)$ and whose landscapes exhibit the familiar exponentially vanishing variance associated with barren plateaus.

\item[(ii)] \textbf{Local product circuits}, generated by $\mathfrak{su}(2)^n$, whose restricted symmetry prevents the concentration of measure and yields an extensive variance.

\item[(iii)] \textbf{Commuting circuits}, generated by $\mathfrak{u}(1)^{2n-1}$, for which the landscape retains finite irreducible contributions, leading to a nonvanishing lower bound on the variance.
\end{itemize}

Together, these examples demonstrate that the trainability of a variational quantum algorithm is governed not simply by the expressibility of the circuit, but by the representation-theoretic structure of its symmetry group.

\subsection{Fully expressive circuits: $\mathfrak{su}(2^n)$}

We begin with the maximally expressive circuit architecture, whose dynamical Lie algebra is $\mathfrak{su}(2^n)$. This example serves as a benchmark for our formalism, since it corresponds to the fully controllable setting previously analyzed within the Lie-algebraic approach~\cite{ragone2024lie}. Consequently, our representation-theoretic framework must reproduce the known Haar-random behavior while illustrating how it emerges naturally from the decomposition into irreducible representation channels.

Let $\varphi$ be the defining representation of $SU(2^n)$ on $\mathcal H_n$ generated by the circuit pictured in Figure \ref{fig:tres_figuras_topo}(a), so that $\varphi$ is irreducible and $|0\rangle^{\otimes n}$ is a highest weight vector \footnote{That corresponds to the DLA $\mathfrak{a}_k$, with $k=12,17,18,19,21,22$, in \cite{wiersema2023classification}.}. Then we have the natural decomposition
\begin{equation}\label{su(2^n)-op_decomp}
    \mathcal B(\mathcal H_n) = \mathrm{span}\{\mathds 1\}\oplus \mathfrak{sl}(2^n)\, ,
\end{equation}
where $\mathfrak{sl}(2^n)$ is the space of traceless operators, which happens to be the complexification of the Lie algebra $\mathfrak{su}(2^n)$. Since $\mathfrak{su}(2^n)$ is simple, the adjoint action of $SU(2^n)$ on $\mathfrak{su}(2^n)$ is an irreducible \emph{real} representation, so the complexification of this representation is irreducible as well. That is to say, \eqref{su(2^n)-op_decomp} realizes the invariant decomposition \eqref{op_inv_decomp}, where $\mathrm{span}\{\mathds 1\}$ is identified with the trivial representation $\mathcal H_\tau$, and $\mathfrak{sl}(2^n)$ is $\mathcal H_\lambda$, for $\lambda$ being the complex adjoint representation \footnote{Note that $m(\tau) = 1$ was already expected since $\varphi$ is irreducible.}.

Thus, for any operator $A\in \mathcal B(\mathcal H_n)$, we have
\begin{equation}
    A = \pi_\tau(A)+\pi_\lambda(A)\, ,
\end{equation}
where
\begin{equation}
    \pi_\tau(A) = \dfrac{\tr(A)}{2^n}\mathds 1\, ,
\end{equation}
so that
\begin{equation}\label{norm_pi_lambda}
\begin{aligned}
    \norm{\pi_\lambda(A)}^2 & = \norm{A}^2-\norm{\pi_\tau(A)}^2\\
    & = \tr(A^\dagger A)-\dfrac{|\tr(A)|^2}{2^n}\, .
\end{aligned}
\end{equation}
In particular, for any initial pure state $\rho$ and observable $O$, the mean given by Theorem \ref{theo:mean_var_2-design} assumes the form
\begin{equation}
    \mathbb E_{\vb*\theta}[\ell_{\vb*\theta}(\rho, O)] = \dfrac{1}{2^n}\tr(O)\, ,
\end{equation}
whereas the variance, by Lemma \ref{lemma:C_zeta_natural}, is given by
\begin{equation}
    \mathrm{Var}_{\vb*\theta}[\ell_{\vb*\theta}(\rho, O)] = \dfrac{1}{4^n-1}\norm{\pi_\lambda(\rho)}^2\norm{\pi_\lambda(O)}^2\, .
\end{equation}
From \eqref{norm_pi_lambda}, we get
\begin{equation}
\mathrm{Var}_{\vb*\theta}[\ell_{\vb*\theta}(\rho, O)] = \frac{2^n - 1}{2^n(4^n - 1)} \left( \tr(O^2) - \frac{\tr(O)^2}{2^n} \right)\, ,
\end{equation}
where we have used that $\rho$ is pure, that is, $\rho^2=\rho$. With this, if we assume that $\tr(O) = 0$, then
\begin{equation}
\begin{aligned}
    & \mathrm{Var}_{\vb*\theta}[\ell_{\vb*\theta}(\rho, O)] = \dfrac{2^n-1}{2^n(4^n-1)}\tr(O^2)\, .
\end{aligned}
\end{equation}
which means that a barren plateau occurs whenever the squared norm of the observable is of order $\mathcal O(a^n)$ with $0<a<4$ because
\begin{equation}
    \lim_{n\to\infty}\dfrac{2^n-1}{2^n} = 1\, .
\end{equation}

For $H$ as in \eqref{ANNNI_ham}, we indeed have $\tr(H) = 0$, and $\tr(H^2)$ is given by \eqref{2-norm_H}, so
\begin{equation}
\begin{aligned}
    \mathrm{Var}_{\vb*\theta}[\ell_{\vb*\theta}(\rho,H)] & = \dfrac{2^n-1}{4^n-1}\left((1+\kappa^2+\gamma^2)n-2\kappa^2-1\right)\\
    & \sim (1+\kappa^2+\gamma^2)\dfrac{n}{2^n}
\end{aligned}
\end{equation}
for any initial pure state $\rho$, which confirms the already expected untrainability of this ansatz.

\subsection{Local product circuits: $\mathfrak{su}(2)^n$}

We now turn to the opposite extreme, namely a circuit architecture that is incapable of generating entanglement. Although considerably less expressive than the previous example, this ansatz provides an illuminating case study because its symmetry fragments the variational landscape into many independent irreducible representation channels. As we shall see, this fragmentation prevents the exponential concentration responsible for barren plateaus.

The relevant symmetry group is $SU(2)^n$, acting independently on each qubit, with $\varphi$ being the representation of $SU(2)^n$ on $\mathcal H_n$ given by
\begin{equation}
    \varphi(g_1,...,g_n) = g_1\otimes...\otimes g_n\, ,
\end{equation}
which is generated by the separable ansatz pictured in Figure \ref{fig:tres_figuras_topo}(b) \footnote{That is generated by the DLA $\mathfrak{b}_3$ in \cite{wiersema2023classification}.}. This is an irrep with $\ket0^{\otimes n}$ as a highest weight vector. The irreps occurring in $\mathcal B(\mathcal H_n)$ are identified by subsets of $[n]:=\{1,...,n\}$ such that for any $S\subset [n]$, the set of Pauli strings
\begin{equation}
    \mathcal P_S = \{2^{-n/2}\sigma_{j_1}\otimes...\otimes \sigma_{j_n}:j_k\ne 0 \iff k \in S\},
\end{equation}
form an orthornomal basis of the space $\mathcal H_S$ carrying a $3^{|S|}$-dimensional irrep of $SU(2)^n$, where distinct subsets give nonequivalent representations. With this, the decomposition \eqref{op_inv_decomp} assumes the form
\begin{equation}
    \mathcal B(\mathcal H_n) = \bigoplus_{S\subset [n]}\mathcal H_{S}\, ,
\end{equation}
where the trivial representation corresponds to the empty set, and
\begin{equation}
    \pi_{S}(A) = \sum_{\sigma \in \mathcal P_S}\dfrac{\tr(\sigma A)}{2^n}\sigma\, .
\end{equation}
Hence
\begin{equation}
    \mathcal C_S(\rho,O) = \dfrac{1}{4^n3^{|S|}}\left(\sum_{\sigma \in \mathcal P_S}\tr(\sigma \rho)^2\right)\left(\sum_{\sigma \in \mathcal P_S}\tr(\sigma O)^2\right)\, .
\end{equation}
In particular, for $H$ as in \eqref{ANNNI_ham}, we have
\begin{equation}
    \sum_{\sigma \in \mathcal P_S}\tr(\sigma H)^2 = \begin{cases}
         4^n\gamma^2\hspace{1.5 em} \mbox{if} \hspace{1 em} |S|=1\\
         4^n \hspace{2.6 em} \mbox{if} \hspace{1 em} S = \{a,a+1\}\\
         4^n\kappa^2 \hspace{1.5 em} \mbox{if} \hspace{1 em} S = \{a,a+2\}\\
         0 \hspace{3.2 em} \mbox{otherwise}
    \end{cases}.
\end{equation}
Now, note that every separable state has the same orbit by $SU(2)^n$. Thus, for $\rho_s$ as in \eqref{state_s} and any separable state $\rho$, we have
\begin{equation}
    \mathcal C_\zeta(\rho, H) = \mathcal C_\zeta(\rho_s, H)\, .
\end{equation}
Explicitly, for $\rho_s$ of \eqref{state_s} and the entangled state $\rho_e$ of $\eqref{state_e}$, it is just a matter of straightforward calculation to get
\begin{equation}
    \sum_{\sigma \in \mathcal P_S}\tr(\sigma \rho_s)^2 = 1 \hspace{2 em} \forall\, S \subset [n]\, ,
\end{equation}
\begin{equation}
    \sum_{\sigma \in \mathcal P_S}\tr(\sigma \rho_e)^2 =\begin{cases}
        0 \hspace{1.5 em} \mbox{if} \hspace{1 em} |S| = 1\\
        1 \hspace{1.5 em} \mbox{if} \hspace{1 em} |S| = 2
    \end{cases}.
\end{equation}
Then, any separable state $\rho$ satisfies
\begin{equation}
\begin{aligned}
    \mathcal C_S(\rho,  H) = \begin{cases}
        \gamma^2/3 \hspace{1.5 em} \mbox{if} \hspace{1 em} |S| = 1\\
        1/9\hspace{2.1 em} \mbox{if} \hspace{1 em} S = \{a,a+1\}\\
        \kappa^2/9 \hspace{1.5 em} \mbox{if} \hspace{1 em} S = \{a,a+2\}\\
        0 \hspace{3.2 em} \mbox{otherwise}
    \end{cases},
\end{aligned}
\end{equation}
whereas the state $\rho_e$ gives
\begin{equation}
    \mathcal C_S(\rho_e,  H) = \begin{cases}
        1/9\hspace{2.1 em} \mbox{if} \hspace{1 em} S = \{a,a+1\}\\
        \kappa^2/9 \hspace{1.5 em} \mbox{if} \hspace{1 em} S = \{a,a+2\}\\
        0 \hspace{3.2 em} \mbox{otherwise}
    \end{cases}.
\end{equation}
Summing up,
\begin{equation}
\begin{aligned}
    \mathrm{Var}_{\vb*\theta}[\ell_{\vb*\theta}(\rho,  H)] = \dfrac{1+\kappa^2+3\gamma^2}{9}\,n-\dfrac{1+2\kappa^2}{9}
\end{aligned}
\end{equation}
for any separable state $\rho$, including the stabilizer state $\rho_s$ and the non-stabilizer state $\rho_m$ of \eqref{state_s}-\eqref{state_s2}, and
\begin{equation}
    \mathrm{Var}_{\vb*\theta}[\ell_{\vb*\theta}(\rho_e,  H)] = \dfrac{1+\kappa^2}{9}\,n-\dfrac{1+2\kappa^2}{9}
\end{equation}
for the maximally entangled state $\rho_e$ of \eqref{state_e}. Interestingly, the variance for an initial GHZ state is independent of the $\gamma$ parameter representing the strength of the transverse field in the ANNNI model. Moreover, between a separable state $\rho$ and the maximally entangled state $\rho_e$, there is no sensible advantage, in terms of trainability, captured by the variances of the cost functions, as both variances scale linearly.

\subsection{Commuting circuits: $\mathfrak{u}(1)^{2n-1}$}

Our final example illustrates the opposite limit of expressibility. Here the dynamical Lie algebra is Abelian, leading to a highly fragmented representation structure in which every irreducible representation is one-dimensional. From this fragmentation, the representation-theoretic framework allows the trainability to be established analytically by identifying a single irreducible sector that necessarily contributes to the landscape variance.

Consider the representation of $\mathfrak u(1)^{2n-1}$ produced by generators $\{X_jX_{j+1},X_j, X_{j+1}:1\le j \le n-1\}$, which is generated by the ansatz pictured in Figure \ref{fig:tres_figuras_topo}(c) \footnote{That is the DLA $\mathfrak{b}_1$ in the classification of \cite{wiersema2023classification}.}. This representation integrates to a unitary representation $\varphi$ of the torus $U(1)^{2n-1}$ on $\mathcal H_n$, and every $\zeta \in \mathfrak{CG}(\varphi)$ is unidimensional.

Each $X_j$ spans a trivial irrep, so the variance for $H$ depends only on the terms carrying interactions between spins. According to Corollary \ref{cor:var_ge_C}, a lower bound for the total variance can be established by evaluating the contribution of any single non-trivial irreducible representation. We can therefore restrict our analysis to a single interaction term from the ANNNI Hamiltonian, such as $Z_1Z_2$, to isolate one such representation. To systematically find its irreducible components, we must decompose this term into operators that act as eigenvectors under the adjoint action of the local generators. With this in mind, note that $Z_1Z_2$ can be written as a linear combination of operators
\begin{equation}
    A_{\epsilon_1,\epsilon_2,\epsilon_3} = \dfrac{1}{2^{(n+3)/2}}(Y_1+\epsilon_1iZ_1)(Y_2+\epsilon_2iZ_2)(\mathds 1+\epsilon_3  X_3)\,,
\end{equation}
for $\epsilon_1,\epsilon_2,\epsilon_3 \in\{-1,1\}$. By computing the commutator of $A_{\epsilon_1,\epsilon_2,\epsilon_3}$ with the generators of the ansatz, one can verify that $A_{\epsilon_1,\epsilon_2,\epsilon_3}$ spans a subspace carrying an irrep whose class is determined by the signs of $\epsilon_1,\epsilon_2,\epsilon_3$. We take $\zeta_0$ to be the irrep identified by $\epsilon_1=\epsilon_2=\epsilon_3 = 1$.

For any Pauli string $\sigma\ne Z_1Z_2$ appearing in the expansion of $H$, we have either $[X_1,\sigma] = 0$ or $[X_2,\sigma]= 0$, however, 
\begin{equation}
    [X_1,A_{1,1,1}] = [X_2,A_{1,1,1}] = 2A_{1,1,1}\, .
\end{equation}
Then
\begin{equation}
    \pi_{\zeta_0}(H) = 2^{(n-3)/2}A_{1,1,1} \, .
\end{equation}
An analogous argument yields
\begin{equation}
    \pi_{\zeta_0}(\rho_s) = -\dfrac{1}{2^{(n+3)/2}}A_{1,1,1}\, .
\end{equation}
Although both $\pi_{\zeta_0}(\rho_m)$ and $\pi_{\zeta_0}(\rho_e)$ have nonzero components orthogonal to $A_{1,1,1}$, for example in the direction of $A_{1,1,1}X_n$, we invoke Lemma \ref{lemma:C_zeta_basis} to conclude they don't matter: setting $A_{1,1,1}$ to be $u(\zeta_0,1;1)$, it follows that
\begin{equation}
    (\rho_m)^{(\zeta_0,1)}_1 = \dfrac{1+\sqrt{2}}{2^{n/2+3}}\, ,
\end{equation}
\begin{equation}
    \ (\rho_e)^{(\zeta_0,1)}_1 = -\dfrac{1}{2^{(n+3)/2}}\, ,
\end{equation}
and, for $j> 1$, $H^{(\zeta_0,j)}_1 = 0$. This is sufficient to conclude that
\begin{equation}
    \mathcal C_{\zeta_0}(\rho_s,H) = \mathcal C_{\zeta_0}(\rho_e,H) = \dfrac{1}{2^6}\, ,
\end{equation}
\begin{equation}
    \mathcal C_{\zeta_0}(\rho_m,H) = \dfrac{3+2\sqrt{2}}{2^9}\, ,
\end{equation}
by Lemma \ref{lemma:C_zeta_basis}. From Corollary \ref{cor:var_ge_C}, we get
\begin{equation}
    \mathrm{Var}_{\vb*\theta}[\ell_{\vb*\theta}(\rho, H)] \ge \mathcal C_{\zeta_0}(\rho, H) = \dfrac{1}{2^6}
\end{equation}
for $\rho$ being either the separable stabilizer state \eqref{state_s} or the maximally entangled state \eqref{state_e}, and
\begin{equation}
    \mathrm{Var}_{\vb*\theta}[\ell_{\vb*\theta}(\rho_m, H)] \ge \mathcal C_{\zeta_0}(\rho_m, H) = \dfrac{3+2\sqrt{2}}{2^9}\, ,
\end{equation}
where $\rho_m$ is the separable non-stabilizer state given by \eqref{state_s2}.

\begin{figure*}[!t]
    \centering
    \includegraphics[width=\textwidth]{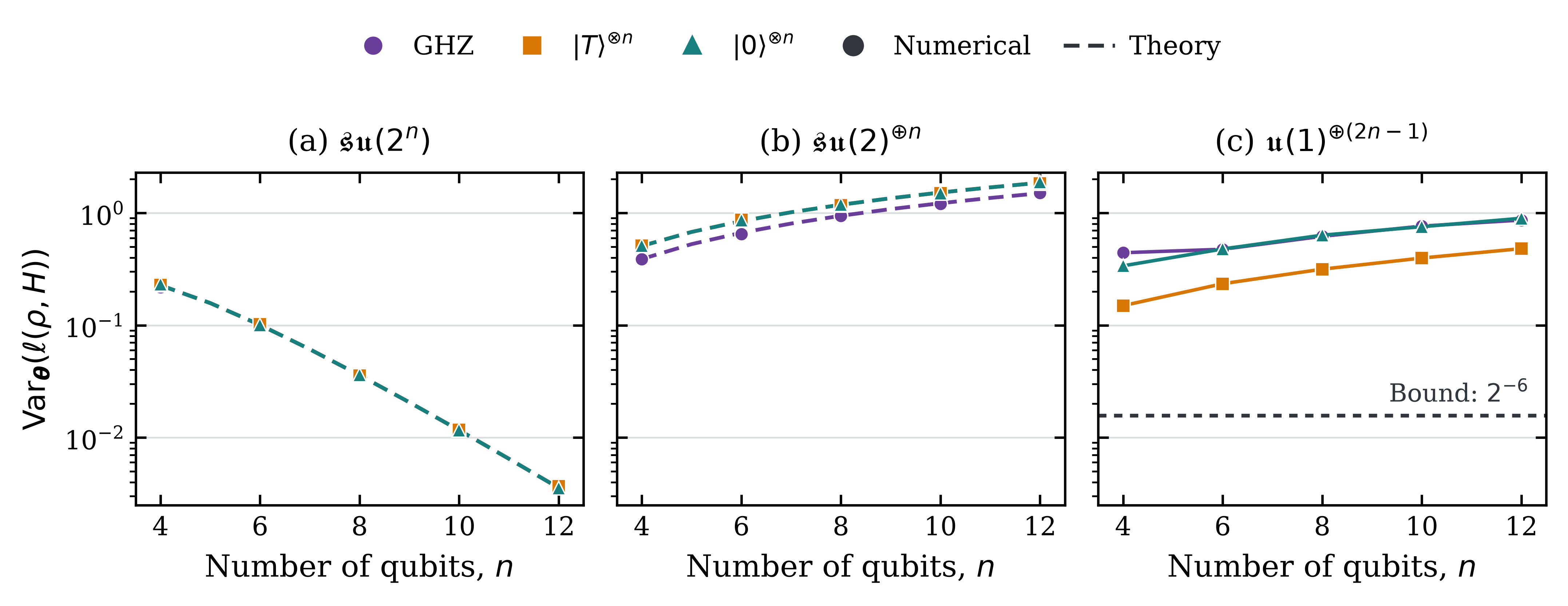}
    \caption{\textbf{Cost-function variance as a function of system size for different ansatz architectures.} Comparison between numerical estimates (markers, sampled over $10^4$ independent realizations) and analytical predictions/bounds (dashed and dotted lines) for initial separable states ($|0\rangle^{\otimes n}$, $|T\rangle^{\otimes n}$) and the maximally entangled state ($|\mathrm{GHZ}\rangle$). 
(a) Fully expressive $\mathfrak{su}(2^n)$ ansatz, exhibiting exponential concentration (barren plateau) independent of the initial state. 
(b) Local $\mathfrak{su}(2)^{n}$ architecture, demonstrating extensive scaling, where separable states and the GHZ state split according to their populated representation sectors. 
(c) $\mathfrak u(1)^{2n-1}$ architecture, showing non-vanishing scaling that remains strictly above the analytical lower bound (dotted line).}
    \label{fig:numericalvar}
\end{figure*}
\subsection{Numerical validation and optimization landscapes}

The analytical results above predict qualitatively distinct trainability regimes for the three circuit architectures. In particular, the variance of the cost function is exponentially suppressed for the fully expressive $\mathfrak{su}(2^n)$ architecture, grows extensively for the local $\mathfrak{su}(2)^n$ architecture, and remains bounded from below for the highly constrained $\mathfrak{u}(1)^{2n-1}$ architecture. We now test these predictions numerically and, in a second set of simulations, investigate whether the corresponding variance scaling is reflected in the actual optimization performance.

\paragraph{Variance scaling.}

We first compare the analytical expressions derived above with direct numerical estimates of the cost-function variance. For each number of qubits and circuit architecture, we sampled $10^4$ independent circuit realizations and evaluated $\mathrm{Var}_{\boldsymbol\theta} [\ell_{\boldsymbol\theta}(\rho,H)]$ across the resulting ensemble. The simulations were performed using the PennyLane library~\cite{bergholm2018pennylane,kwak2021introduction}. We considered both the separable initial states $\rho_s$ and $\rho_m$, as well as the maximally entangled GHZ state $\rho_e$, whenever the analytical results distinguish between them.

Figure~\ref{fig:numericalvar} compares these numerical estimates with the analytical predictions and bounds obtained from our representation-theoretic decomposition. Across all three architectures, the agreement is excellent. For the fully expressive $\mathfrak{su}(2^n)$ ansatz, the numerical variance follows the predicted exponential suppression (Fig.~\ref{fig:numericalvar}(a)), providing a direct signature of a barren plateau. In contrast, the $\mathfrak{su}(2)^n$ architecture exhibits the predicted extensive scaling (Fig.~\ref{fig:numericalvar}(b)), where the separable and GHZ states differ only through the representation sectors in which their operator content has support. This non-vanishing behavior is particularly notable, as cost-function variances previously reported in the literature decay with system size. Finally, for the $\mathfrak{u}(1)^{2n-1}$ architecture (Fig.~\ref{fig:numericalvar}(c)), the numerical variance increases linearly and remains above the analytical lower bound, as required by Corollary~\ref{cor:var_ge_C}.

These results provide direct numerical validation of the central prediction of our framework: the scaling of landscape fluctuations is determined by the irreducible representation sectors jointly populated by the initial state and the cost observable. Importantly, the three architectures exhibit qualitatively and quantitatively distinct behaviors despite operating on the same Hamiltonian. The numerical results therefore support the conclusion that the symmetry structure of the ansatz, rather than the Hamiltonian alone, dictates the concentration properties of the variational landscape.

\paragraph{Optimization performance.}

While cost-function variance quantifies the trainability of the landscape, its practical significance is operational: non-vanishing fluctuations must enable an optimizer to reliably navigate the parameter space. To examine this connection, we evaluate the variational minimization of a $10$-qubit ANNNI Hamiltonian with couplings $\kappa = 0.5$ and $\gamma = 1.0$ (ground-state energy $E_0 \approx -11.79$) and initial state $\rho_s$ as in \eqref{state_s}. We track the optimization performance across circuit depths $L \in \{1, \dots, 10\}$ for the same three representative architectures, namely, $\mathfrak{su}(2^n)$, $\mathfrak{su}(2)^n$ and $\mathfrak{u}(1)^{2n-1}$.

First, as shown in Fig.~\ref{fig:optimizationf}(a), the distance to a unitary 2-design for the $\mathfrak{su}(2^n)$ architecture steadily decreases as the number of layers increases, crossing the $\varepsilon = 10^{-3}$ approximate 2-design threshold around $L = 10$. This convergence can be estimated analytically, as demonstrated in Ref.~\cite{ragone2024lie}. The corresponding optimization performance across various circuit depths is presented in Fig.~\ref{fig:optimizationf}(b). To evaluate these trajectories, we employed the parameter-shift rule for 250 steps per optimization, repeating the process 12 times.

For the fully expressive $\mathfrak{su}(2^n)$ architecture, performance degrades significantly as the circuit depth increases beyond $L = 5$. As the circuit grows and the unitaries enter the approximate 2-design regime, the exponential suppression of variance manifests as a severely flattened landscape. This leads to optimization bottlenecks, evidenced by the large standard deviations (shaded regions) across independent runs and the worsening of the final achieved energy at $L = 10$. In contrast, the symmetrically restricted $\mathfrak{su}(2)^n$ and $\mathfrak u(1)^{2n-1}$ architectures do not suffer from this exponential concentration. They retain substantial landscape fluctuations and maintain stable, reliable optimization performance across all evaluated depths without falling into barren plateaus. Consequently, despite the well-known trainability challenges associated with deep variational quantum circuits, constraining the ansatz can yield robust optimizations. Interestingly, although the $\mathfrak{su}(2)^n$ architecture restricts the system entirely to separable states, it still achieves a surprisingly accurate estimate of the ground-state energy.

Comparing the variance scalings with the optimization performance in Fig.~\ref{fig:optimizationf}(b) bridges abstract representation theory and practical trainability. The former establishes our predicted fluctuation bounds, while the latter demonstrates their operational consequence: architectures whose relevant representation sectors preserve non-vanishing variance sustain actionable gradient signals. In contrast, the exponential concentration inherent to fully expressive circuits dictates an inevitable barren plateau as depth increases. Ultimately, these results visually confirm that constraining ansatz expressivity via Lie-algebraic symmetry fundamentally protects trainability in deep quantum circuits.
\begin{figure*}[t!] 
    \centering
   \begin{subfigure}[b]{0.47\textwidth}
        \centering
        \includegraphics[width=\textwidth]{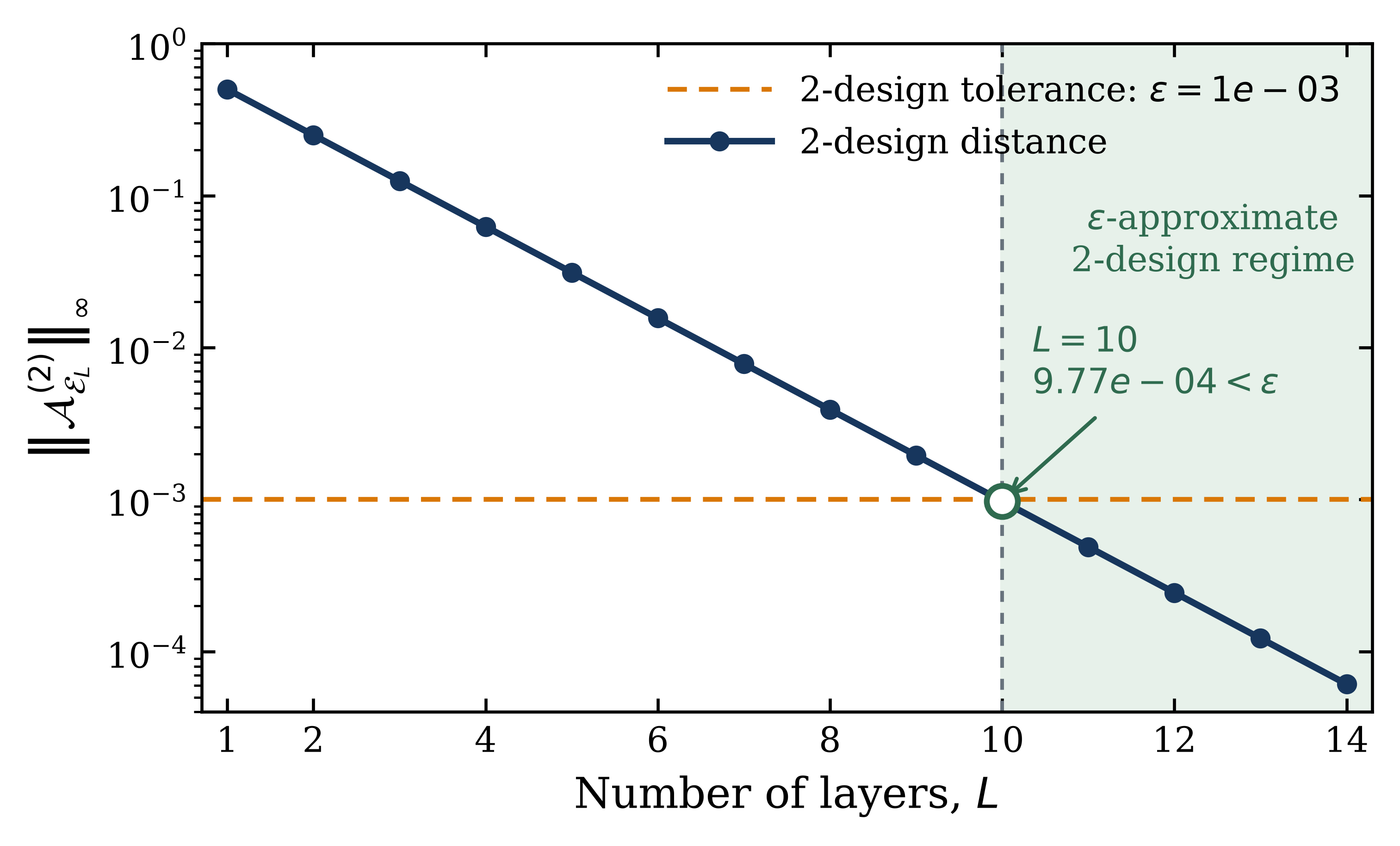}
        \caption{}
        \label{fig:scaling_a}
    \end{subfigure}
    \hfill
    \begin{subfigure}[b]{0.51\textwidth}
        \centering
        \includegraphics[width=\textwidth]{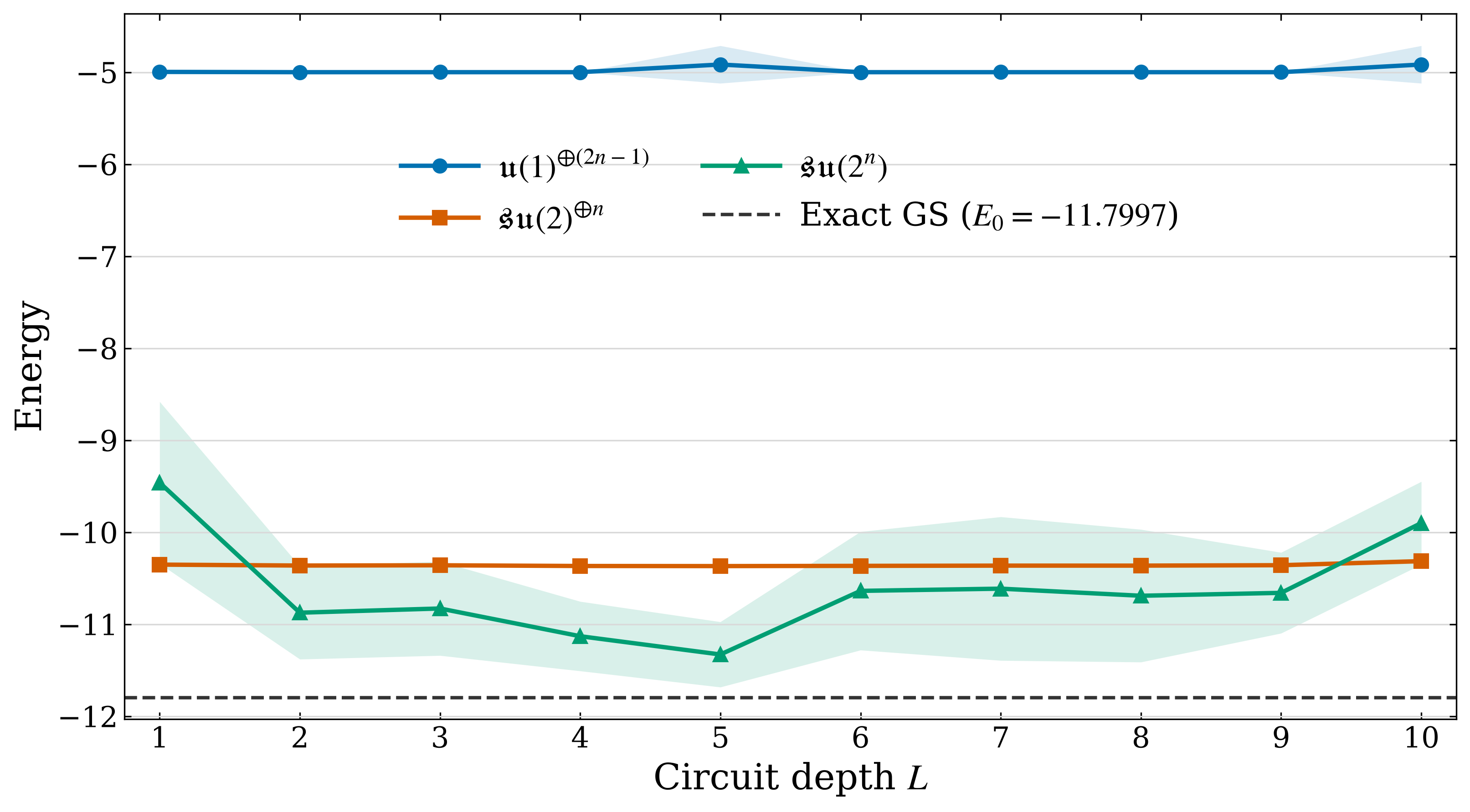}
        \caption{}
        \label{fig:landscape_c}
    \end{subfigure}
    \caption{\textbf{Convergence to a 2-design and optimization performance.} 
(a) Distance to a unitary 2-design ($\mathcal{A}_{\epsilon}^{(2)}$) as a function of the number of layers $L$ for the $SU(2^n)$ architecture. As circuit depth increases, the ansatz ensemble approaches an approximate 2-design regime, crossing the tolerance threshold $\epsilon = 10^{-3}$ around $L = 10$. 
(b) Variational ground-state energy for the 10-qubit ANNNI model ($\kappa=0.5$, $\gamma=1.0$) across circuit depths $L \in \{1, \dots, 10\}$ comparing the three architectures: $\mathfrak{su}(2^n)$, $\mathfrak{su}(2)^n$, and $\mathfrak u(1)^{2n-1}$. The horizontal dashed line denotes the exact ground-state energy ($E_0 \approx -11.7997$). As $\mathfrak{su}(2^n)$ approaches 2-design statistics, exponential concentration leads to optimization bottlenecks, whereas $\mathfrak{su}(2)^n$ and $\mathfrak u(1)^{2n-1}$ retain non-vanishing fluctuations and reliably achieve lower energies.}
    \label{fig:optimizationf}
\end{figure*}

\section{Discussion}

In this work, we developed a representation-theoretic framework for analyzing fluctuations of variational quantum landscapes generated by deep parameterized quantum circuits. By viewing the cost function as a function on the orbit of a compact Lie group, we recast the characterization of landscape fluctuations as a problem in harmonic analysis. The resulting decomposition into irreducible representation channels, together with Schur orthogonality, allows us to derive exact expressions and analytical bounds for the cost-function variance without requiring either the initial state or the observable to belong to the dynamical Lie algebra. This substantially extends the scope of existing Lie-algebraic approaches to barren plateaus.

A central feature of the framework is that it identifies the irreducible representation channels as the fundamental carriers of landscape fluctuations. The contribution of each channel is determined jointly by the overlap of the initial state and the observable with that sector. Our central result, Theorem \ref{theo:mean_var_2-design}, provides exact analytical expressions alongside upper and lower bounds for the cost function variance across arbitrary irreducible representation channels, $\mathcal{C}_{\zeta}(\rho, O)$. The Lie-algebraic formulas of Ref.~\cite{ragone2024lie} are recovered as a special case when the relevant operator content is restricted to the adjoint representation, demonstrating that the established theory is naturally embedded within the more general representation-theoretic description developed here.

The applications of the ANNNI model illustrate how this perspective translates into qualitatively different regimes of trainability. For the fully expressive $\mathfrak{su}(2^n)$ architecture, the variance is exponentially suppressed, recovering the conventional barren plateau. In contrast, the $\mathfrak{su}(2)^n$ architecture exhibits an extensive variance, while the highly constrained $\mathfrak{u}(1)^{2n-1}$ architecture admits a system-size-independent lower bound. These examples show that trainability is not determined solely by the amount of expressibility or entanglement generated by a circuit. Rather, it is controlled by how the cost function and the initial state are distributed across the irreducible representation sectors of the circuit's symmetry group. In this sense, the representation structure provides a direct link between the symmetries of a variational ansatz and the geometry of its optimization landscape.

Our results also place the present framework in the broader development of barren plateau theory. The dynamical Lie algebra approach of Ref.~\cite{ragone2024lie} established exact variance formulas under Lie-algebraic assumptions, while Ref. \cite{fontana2024characterizing} clarified the role of the adjoint representation in this description. Further, Ref. \cite{diaz2023showcasing} subsequently demonstrated that the DLA restriction can be overcome in specific settings by exploiting additional representation structure. The framework developed here takes this progression further by treating the full irreducible representation content of the underlying compact group as the organizing principle. From this perspective, these seemingly different approaches correspond to particular sectors or special cases of a common representation-theoretic description.

Thus, our framework provides the theoretical tools necessary to systematically analyze symmetry-tailored architectures \cite{PRXQuantum.3.030341,PRXQuantum.5.020328}. By explicitly mapping how a restricted group action decomposes into irreducible representation channels, it becomes possible to rigorously evaluate whether such ansatze will scale favorably in practice or if they will succumb to generalized entanglement or locality issues. However, the present analysis relies on the assumption that the variational circuit realizes a unitary $2$-design, which provides the averaging structure required for our exact variance formulas. Although this assumption captures the asymptotic regime of sufficiently deep and expressive circuits, it can be restrictive for finite-depth architectures, where the approach to a $2$-design may itself determine the practical onset of barren plateaus. Extending the present framework to quantify finite-depth corrections and to characterize landscapes away from the $2$-design regime is therefore an important direction for future work.

More broadly, our results suggest that harmonic analysis provides a natural language for studying variational quantum landscapes beyond barren plateaus themselves. Rather than viewing the landscape solely through the algebra of the circuit generators, one can resolve its fluctuations into symmetry-adapted representation channels and ask which sectors survive, dominate, or become suppressed as the system size grows. This perspective opens the possibility of using representation-theoretic structure not only to diagnose barren plateaus, but also to guide the design of variational architectures whose symmetries preserve trainable optimization landscapes.

\acknowledgements
We acknowledge financial support the Simons Foundation (Grant No. 1023171, R.C.), the Brazilian National Council for Scientific and Technological Development (CNPq, Grants No. 403181/2024-0, 301687/2025-0 and 150631/2026-0), the National Institute of Science and Technology for Applied Quantum Computing through CNPq process No. 408884/2024-0, the Financiadora de Estudos e Projetos (Grant No. 1699/24 IIF-FINEP). RC thanks the Technical University of Denmark for its hospitality, where part of this work was carried out during a guest professorship supported by the Otto M\o nsted Foundation. We also thank the High-Performance Computing Center (NPAD) at UFRN for providing computational resources. The authors acknowledge the use of Google Gemini for assistance in styling and rendering the illustrations in Figures 1 and 2. 

\bibliography{references}
\newpage
\appendix

\section{Proofs of results in Section \ref{sec:framework}}\label{sec:proof}

Except for the corollaries, here we restate the results of Section \ref{sec:framework}, but now adding their proofs. We'll repeatedly invoke Schur's Orthogonality Relations, and we refer to \cite{folland} for its statement and proof.

\begin{lemma*}\label{lemma:mean_beta-ap}
\begin{equation*}
    \langle \beta_{\rho, O}\rangle = \ip{\pi_\tau(\rho)}{\pi_\tau(O)} = \tr(\pi_\tau(\rho)\pi_\tau(O))\, .
\end{equation*}
\end{lemma*}
\begin{proof}
    By definition,
    \begin{equation}
        \widetilde \beta_{\rho, O}(g) = \ip{\rho^g}{O}\, ,
    \end{equation}
 so the result follows by applying Schur's Orthogonality Relations to the decomposition \eqref{op_inv_decomp}. The group average annihilates every nontrivial irreducible component, leaving only the components carrying the trivial representation.
\end{proof}

\begin{proposition*}
\label{prop2-ap}
    Under the hypothesis of $2$-design, if $\varphi$ is an irrep and $\tr(O) = 0$, then $\mathbb E_{\vb*\theta}[\ell_{\vb*\theta}(\rho,O)] = 0$ and $\mathrm{Var}_{\vb*\theta}[\ell_{\vb*\theta}(\rho, O)]=\langle\beta_{\rho, O}^2\rangle$.
\end{proposition*}
\begin{proof}
    By Schur's Lemma, there is only one copy of $\tau$, carried by the span of $\mathds 1$, because $\varphi$ is an irrep. The hypothesis $\tr(O) = 0$ implies that $O \perp \mathds 1$, then $\langle\beta_{\rho, O}\rangle = 0$ by the previous lemma. 
\end{proof}

\begin{lemma*}\label{lemma:mean_beta_sq-ap}
    \begin{equation*}
    \begin{aligned}
        \langle \beta_{\rho, O}^2\rangle & = \sum_{\zeta\in\mathfrak{CG}(\varphi)}\mathcal C_\zeta(\rho, O)\, .
    \end{aligned}
    \end{equation*}
\end{lemma*}
\begin{proof}
Since both $\rho$ and $O$ are Hermitian operators, we have that $\widetilde \beta_{\rho, O}(g)$ coincides with its complex conjugation, so \eqref{Beta_projection} gives
\begin{equation}\label{lift_beta_sq}
    \widetilde\beta_{\rho, O}(g)^2 = \sum_{\zeta,\zeta'\in \mathfrak{CG}(\varphi)}\!\ip{\pi_\zeta(\rho)^g}{O}\ip{O}{\pi_{\zeta'}(\rho)^g}.
\end{equation}
The second moment therefore consists of pairwise correlations between representation sectors. To isolate these contributions, we introduce the family of operators $T_\rho^{\zeta,\zeta'}$: for $\zeta,\zeta' \in \mathfrak{CG}(\varphi)$, let $T_\rho^{\zeta,\zeta'}:\mathcal B(\mathcal H_n)\to \mathcal B(\mathcal H_n)$ be the linear operator given by
\begin{equation}\label{T_op}
    T_\rho^{\zeta,\zeta'}(A) := \int_G\ip{\pi_\zeta(\rho)^g}{A}\pi_{\zeta'}(\rho)^g dg \, ,
\end{equation}
so that
\begin{equation}
    \langle\beta_{\rho, O}^2\rangle = \sum_{\zeta,\zeta'}\ip{O}{T_\rho^{\zeta,\zeta'}(O)}\, .
\end{equation}
The next step is to understand the structure of the operators $T_\rho^{\zeta,\zeta'}$. 

The first thing to know is that $T_\rho^{\zeta,\zeta'}$ commutes with the action \eqref{op_action}. Indeed, the invariance of the Haar integral straightforwardly implies that
\begin{equation}
     T_\rho^{\zeta,\zeta'}(A^h) = T_\rho^{\zeta,\zeta'}(A)^h
\end{equation}
for every $h \in G$.

Moreover, it is immediate from the definition that $T_\rho^{\zeta,\zeta'}$ vanishes when applied to any operator orthogonal to the $\zeta$-sector, and its image lies in the $\zeta'$-sector. In other words, $\ker(\pi_\zeta)\subset \ker(T_\rho^{\zeta,\zeta'})$ and $\mathrm{im}(T_\rho^{\zeta,\zeta'})\subset \mathrm{im}(\pi_{\zeta'})$, so it can be seen as an operator from the $\zeta$-sector to the $\zeta'$-sector. 

Finally, we conclude that $T_\rho^{\zeta,\zeta'}$ is identically null if $\zeta\ne \zeta'$. This is so because it can be written as a direct sum of equivariant maps between non-equivalent representations by taking the orthogonal decomposition
\begin{equation}
    \mathcal H_\zeta \otimes \mathbb C^{m(\zeta)} = \bigoplus_{j=1}^{m(\zeta)}\mathcal H_\zeta\otimes e_j
\end{equation}
of the $\zeta$-sector, cf. \eqref{inv_basis}-\eqref{proj_psi_decomp}. The claim, then, follows from Schur's Lemma.

Consequently, there is no cross contribution to $\langle \beta_{\rho,O}^2\rangle$ from different irreducible representation sectors; only contributions with $\zeta=\zeta'$ survive when we integrate \eqref{lift_beta_sq}. By construction, the surviving contribution from $\zeta$ is precisely $\mathcal C_\zeta(\rho, O)$ as in \eqref{psi_comp_mean_beta_sq}.
\end{proof}

\begin{lemma*}\label{lemma:C_tau-ap}
For the trivial representation,
\begin{equation*}
    \mathcal C_\tau(\rho, O) = \tr(\pi_\tau(\rho)\pi_\tau(O))^2\, . 
\end{equation*}
\end{lemma*}
\begin{proof}
    It follows from $\pi_\tau(\rho)^g = \pi_\tau(\rho)$ and normalization of Haar measure.
\end{proof}

\begin{lemma*}\label{lemma:C_zeta_natural-ap}
For any $\zeta \in \mathfrak{CG}(\varphi)$, we have 
\begin{equation*}
    0\le \mathcal C_\zeta(\rho, O) \le \dfrac{1}{\dim \zeta}\norm{\pi_\zeta(\rho)}^2\norm{\pi_\zeta(O)}^2\, ,
\end{equation*}
where the rightmost inequality is an actual equation if $\zeta$ is multiplicity free.
\end{lemma*}
\begin{proof}
By definition,
\begin{equation}
    C_\zeta(\rho, O) = \ip{O}{T^{\zeta,\zeta}_\rho(O)} = \int_G|\ip{\pi_\zeta(\rho)^g}{O}|^2dg\, ,
\end{equation}
and the integrand is a nonnegative function, so the integral is nonnegative as well (this reflects the fact that $T^{\zeta,\zeta}_\rho$ is a positive operator). Also, by Schur's Lemma, the operator $T^{\zeta,\zeta}_\rho$ is diagonalizable, so its operator norm is its highest eigenvalue, which is, in turn, bounded above by $\tr(T_\rho^{\zeta,\zeta})/\dim\zeta$, hence
\begin{equation}
    \mathcal C_\zeta(\rho, O) \le \norm{\pi_\zeta(O)}^2\tr(T^{\zeta,\zeta}_\rho)/\dim\zeta\, .
\end{equation}
Now, note that
\begin{equation}
    \tr(T^{\zeta,\zeta}_\rho) = \int_G\ip{\pi_\zeta(\rho)^g}{\pi_\zeta(\rho)^g}dg = \norm{\pi_\zeta(\rho)}^2\, ,
\end{equation}
yielding the upper bound in the statement. If $\zeta$ is multiplicity free, we have
\begin{equation}
    T^{\zeta,\zeta}_\rho = \dfrac{\tr(T^{\zeta,\zeta}_\rho)}{\dim \zeta}\mathds 1\, ,
\end{equation}
so the equation is given by Schur's Orthogonality Relations.
\end{proof}

\

\begin{lemma*}\label{lemma:C_zeta_basis-ap}
For \eqref{inv_basis}-\eqref{proj_psi_decomp},
    \begin{equation*}
        \mathcal C_\zeta(\rho, O) = \dfrac{1}{\dim \zeta}\sum_{k,l=1}^{\dim\zeta}\left|\sum_{j=1}^{\mathfrak m(\zeta)}O^{(\zeta,j)}_{k}\overline{\rho^{(\zeta,j)}_l}\right|^2\, .
    \end{equation*}
\end{lemma*}
\begin{proof}
The claim follows by using \eqref{proj_psi_decomp} in \eqref{psi_comp_mean_beta_sq} and invoking, once again, Schur's Orthogonality Relations.
\end{proof}

\begin{theorem*}\label{theo:mean_var_2-design-ap}
    Under the hypothesis of $2$-design, the mean and the variance of $\ell_{\vb*\theta}(\rho, O)$ are given by
    \begin{equation*}
        \mathbb E_{\vb*\theta}[\ell_{\vb*\theta}(\rho,O)] = \tr(\pi_\tau(O)\pi_\tau(\rho))\, ,
    \end{equation*}
    \begin{equation*}
    \begin{aligned}
        \mathrm{Var}_{\vb*\theta}[\ell_{\vb*\theta}(\rho, O)]& = \sum_{\substack{\zeta\in \mathfrak{CG}(\varphi)\\ \zeta\ne \tau}}\mathcal C_\zeta(\rho, O)\\
        & \le \sum_{\substack{\zeta\in \mathfrak{CG}(\varphi)\\ \zeta\ne \tau}}\dfrac{1}{\dim\zeta}\norm{\pi_\zeta(\rho)}^2\norm{\pi_\zeta(O)}^2\, .
    \end{aligned}
    \end{equation*}
\end{theorem*}
\begin{proof}
    Recall \eqref{mean_loss}-\eqref{var_loss}. The expression for $\mathbb E_{\vb*\theta}(\ell_{\vb*\theta}(\rho,O))$ is just Lemma \ref{lemma:mean_beta-ap}. The expression for $\mathrm{Var}_{\vb*\theta}[\ell_{\vb*\theta}(\rho, O)]$ follows from Lemmas \ref{lemma:mean_beta_sq-ap}-\ref{lemma:C_zeta_natural-ap}.
\end{proof}

\begin{lemma*}\label{lemma:bound_sch_norm-ap}
If $\varphi$ is an irrep, then
    \begin{equation*}
        \langle \beta_{\rho,O}^2\rangle\le \dfrac{\norm{O}_\infty\norm{O}_1}{2^n}\, .
    \end{equation*}
\end{lemma*}
\begin{proof}
    First, note that
    \begin{equation}
        |\beta_{\rho,O}(g)|\le \beta_{\rho, |O|}(g)\, ,
    \end{equation}
    hence
    \begin{equation}
        \langle \beta_{\rho,O}^2\rangle \le \langle \beta_{\rho,|O|}^2\rangle\, .
    \end{equation}
    Now, let $\widetilde O = \norm{O}_\infty^{-1}|O|$. By definition, $0\le \widetilde O\le \mathds 1$, so $\beta_{\rho, \widetilde O}(g)^2\le \beta_{\rho,\widetilde O}(g)$. Therefore
    \begin{equation}
        \langle \beta_{\rho, \widetilde O}^2\rangle\le \langle \beta_{\rho, \widetilde O}\rangle\,= \dfrac{1}{2^n}\tr(\widetilde O) = \dfrac{\norm{O}_\infty^{-1}\norm{O}_1}{2^n}.
    \end{equation}
    But
    \begin{equation}
        \langle\beta_{\rho, |O|}^2\rangle = \norm{O}_\infty^2\langle\beta_{\rho, \widetilde O}^2\rangle\, .
    \end{equation}
    Putting everything together, we get the statement.
\end{proof}

\begin{theorem*}\label{theo:bound_sch_norm-ap}
    Under the hypothesis of $2$-design, if $\varphi$ is an irrep, then the variance of $\ell_{\vb*\theta}(\rho, O)$ satisfies
    \begin{equation}
        \mathrm{Var}_{\vb*\theta}[\ell_{\vb*\theta}(\rho, O)]\le \dfrac{\norm{O}_\infty\norm{O}_1-\tr(O)^2}{2^n}\, .
    \end{equation}
\end{theorem*}
\begin{proof}
    It follows from Lemmas \ref{lemma:mean_beta-ap} and \ref{lemma:bound_sch_norm-ap}
\end{proof}

\end{document}